\documentclass{article}
\usepackage{ijcai25}

\usepackage{times}
\usepackage{soul}
\usepackage{url}
\usepackage[hidelinks]{hyperref}
\usepackage[utf8]{inputenc}
\usepackage[small]{caption}
\usepackage{graphicx,comment}
\usepackage{amsmath,amsfonts}
\usepackage{amsthm}
\usepackage{booktabs}
\usepackage{algorithm}
\usepackage{algorithmic}
\usepackage[switch]{lineno}

\newtheorem{theorem}{Theorem}
\newtheorem{lemma}{Lemma}

\newtheorem{definition}{Definition}

\title{\LARGE \bf On the Power of Spatial Locality on Online Routing Problems}

\author{
Swapnil Guragain,
Gokarna Sharma\\
\affiliations
Kent State University\\
\emails
\{sguragai, gsharma2\}@kent.edu
}
\begin{document}

\maketitle

\begin{abstract}
We consider the online versions of two fundamental routing problems, traveling salesman (TSP) and dial-a-ride (DARP), which have a variety of relevant applications in logistics and robotics. The online versions of these problems concern with efficiently serving a sequence of requests presented in a real-time on-line
fashion located at points of a metric space by servers (salesmen/vehicles/robots). 
In this paper, motivated from real-world applications, such as Uber/Lyft rides, where some limited knowledge is available on the future requests, we propose the {\em spatial locality} model that provides in advance the distance within which new request(s) will be released from the current position of server(s). We study the usefulness of this advanced information on achieving the improved competitive ratios for both the problems with $k\geq 1$ servers, compared to the competitive results established in the literature without such spatial locality consideration. We show that small locality is indeed useful in obtaining improved competitive ratios irrespective of the metric space. 
\end{abstract}


\section{Introduction}
We consider the {\em online} versions of two fundamental routing problems, {\em traveling salesman} ({\sc Tsp})  and {\em dial-a-ride} ({\sc Darp}), which we denote as {\sc oTsp} and {\sc oDarp}. 
These online versions concern with efficiently serving a sequence of requests presented in a real-time online
fashion located at points of a metric space. Real-time online means that requests arrive at arbitrary times, not necessarily sequentially 
and not necessarily a single request at any time step\footnote{This differentiates this online model from the other online models where requests are assumed to arrive over time in a sequential order, i.e., next request arrives only after  current request has been served. Additionally, only a single request arrives at a time.}.  These kind of real-time online problems have a variety of relevant applications in logistics and robotics. 
Example applications include a set of salesmen/repairmen/vehicles/robots that have to serve locations on its workspace (e.g.,  Euclidean plane) and many other routing and scheduling
problems on a transportation network. 
In this paper, we refer to as {\em servers} the salesmen/repairmen/vehicles/robots. A request arrived at a location is considered {\em served} when a server reaches that location following a route. As the input is communicated to the servers in a real-time online fashion, the scheduled routes will have to be updated also in a
real-time online fashion during the trips of the servers. Since online execution is unaware of future requests, designing an optimal schedule (with respect to relevant performance metrics) is generally not possible and the goal is to reach as close to optimal as possible. 

In  {\sc oDarp}, there are $m$ (not known beforehand) ride requests arriving over time (hence online).  Each ride request consists of the corresponding source and destination points in the metric space that satisfies triangle inequality (which we call {\em arbitrary metric} throughout the paper). We also consider a special case in which the arbitrary metric space becomes a line, which we call {\em line metric}. A ride request is considered {\em served} if a server visits first the source point and then the destination point. There are $k\geq 1$ servers initially positioned at a {\em distinguished origin}, a unique point in the metric space known to servers. Each server travels with unit speed, meaning that it traverses unit distance in unit time. We consider the {\em uncapacitated} variant of {\sc oDarp} meaning that a server can serve simultaneously as many requests as possible. 
%
 {\sc oTsp} is similar to {\sc oDarp} where the source and destination points coincide, i.e., reaching to source point serves the request. In certain situations, the server(s) need to return to the distinguished origin after serving all the requests. If such a requirement, we refer as {\em homing} (or closed),   {\em nomadic} (or open) otherwise. 

In the literature,
the {\em offline} versions of both {\sc Tsp} and {\sc Darp} were extensively studied. Since complete input is known beforehand, the offline schedules are relatively easier to design. However, the offline versions cannot capture the real-world situation in which the complete input may not be available to the server(s) a priori. An online algorithm makes a decision on what to do based on the input (or requests) known so far (without knowledge on future requests). 
A bit formally, let $\sigma_t\subseteq \sigma$ denote a request sequence (or input) arrived over time until time $t\geq 0$ from the set $\sigma$ (not known a priori). In the offline case, all the requests arrive at time $0$ and hence all the requests are known at time $0$, i.e., $\sigma_t=\sigma$ and additionally, the size of the sets $|\sigma_t|=|\sigma|$. However, in the online case, at any time $t\geq 0$, it is not known whether $|\sigma_t|=|\sigma|$. The only thing that is known is all the request that arrived until $t$. We denote by $OL(\sigma_t)$ the cost of the online algorithm at time $t$ to serve the requests that arrived in the system until time $t$. We denote by $OPT(\sigma)$ is the cost of the offline algorithm $OPT$ to serve the requests in $\sigma$ assuming that $OPT$ knows $\sigma$ completely at time $t=0$. 
A standard technique to  evaluate the quality of a solution provided by an online algorithm is {\em competitive analysis}.
An online algorithm $OL$ is $c$-competitive if for any request sequence $\sigma_t\subseteq \sigma$ upto time $t\geq 0$ it holds that $OL(\sigma_t)\leq c\cdot  OPT(\sigma)$. 

The execution proceeds in time steps. 
Suppose a request $r_i$ arrives (releases) at time $t_i\geq 0$. In the {\em original online model}, it is assumed that the request  $r_i$ arrived (released) at time $t_i$ is known at time $t_i$.  That is, no server can serve $r_i$ before $t_i$. Let $r_i$ be served by a server at time $t_i'$. Time $t_i'$ denotes the {\em completion time} of $r_i$. We have that $t_i'$ cannot be smaller than $t_i$, i.e., $t_i'\geq t_i$. Notice also that, if a request arrived at $t_i$ can be served at $t_i'=t_i$, then that algorithm would be best possible (i.e., 1-competitive).

The goal in both {\sc oTSP} and {\sc oDarp} 
is to minimize the {\em maximum} completion time $\max_{1\leq i\leq m}t_i'$ for all $m$ (not known a priori) requests in $\sigma$ (i.e., $|\sigma|=m$). 
%
%
In the original online model,
 for nomadic {\sc oTsp}, 
 the lower bound is 2.04 in line metric \cite{Bjelde20} and upper bound is $1+\sqrt{2}=2.41$ \cite{lipmann2003line} in arbitrary metric. For nomadic {\sc oDarp}, the lower bound is 2.0585 in line metric and upper bound is $2.457$ \cite{BaligacsDSW23} in arbitrary metric.  
 For homing versions of both problems, tight competitive ratios of 2 were obtained. We mean by `tight' here that the competitive ratio is 2 as both lower bound and upper bound.
 The significance of these previous studies is that 
 they were achieved with being oblivious to any knowledge on future requests, such as, at what time and at which location the future requests will arrive as well as how many of such future requests are there in $\sigma$. 
 Although designing algorithms that do not rely on any knowledge on future requests has merit on itself, these algorithms do not exploit the knowledge, if available, on future requests in many  real-world applications. 

Consider for example Uber/Lyft rides as a real-world applications in which some limited knowledge is available on the future requests, i.e., Uber/Lyft drivers 
are provided by the Uber/Lyft ride assignment system with ride requests to serve within certain radius of their current positions. Additionally, Uber/Lyft drivers may typically set the radius within which they plan to accept the ride requests.  
However, the previous studies on the original model fail to capture such knowledge. 
Therefore, the following question naturally arises: {\it Can providing an algorithm with limited clairvoyance, i.e., the capability to foresee some limited knowledge on future requests, help in achieving  a better competitive ratio?} In other words, is limited clairvoyance helpful in reducing the gap between competitive ratio upper and lower bounds?

\begin{table*}[!t]
\footnotesize
    \centering
    \begin{tabular}{cccccc}
    \toprule
    {\bf Algorithm} & {\bf Nomadic {\sc oTsp}} &{\bf Homing {\sc oTsp}} & {\bf Nomadic {\sc oDarp}} & {\bf Homing {\sc oDarp}} & {\bf Metric} \\
    
     & ({\bf lower/upper}) & ({\bf lower/upper}) & ({\bf lower/upper}) & ({\bf lower/upper}) &   \\
     \hline 
     & & {\bf Single server} & & &\\
    \toprule
    Original  & $2.04/2.04$ $^4$ & $1.64/1.64 $ $^{4,6}$ 
    & 2.0585/2.457 $^{7,8}$ & 2/2 $^{9,10}$ &line \\
    \hline
    Original & 2.04/2.41 $^{4,5}$ 
    & 2/2 $^6$  & 2.0585/2.457 $^{7,8}$  & 2/2  $^{9,10}$ &arbitrary \\
    \hline
    Lookahead$^1$ 
    & -/$(1+\frac{2}{\alpha})$ & -/$(1+\frac{2}{\alpha})$ & -/-  & -/-  &line \\
    \hline
    Lookahead$^1$ 
    & 2/$\max\{2,$ & 2/2 & -/-  & -/-  &arbitrary \\
    & $1+\frac{1}{2}(\sqrt{\alpha^2+8}-\alpha)\}$&&&& \\
    \hline
    Disclosure$^2$ 
    & -/- & -/$(2-\frac{\rho}{1+\rho})$ & -/-  & -/-  &arbitrary \\
    \hline
    Temporal$^{11}$ & -/$\min\{2.04,1+\frac{3}{2\beta'\delta'}\}$ & -/$\min \{2,$ & -/$\min\{4,1+\frac{3}{\beta'\delta'}\}$  & -/$\min \{3,$  &line \\
    & & $1+\frac{2}{\max\{2,\beta'\delta'\}}\}$ & & $1+\frac{4}{\max\{2,\beta'\delta'\}}\}$ & \\
    \hline
    Temporal & 2/$\min\{2.41,2+\delta'\}$ & 2/2 & 2/$\min\{2.457,2+\delta'\}$  & 2/2  &arbitrary \\
    
    \toprule\toprule
    {\bf Spatial} & -/$\min\{2.04,1+\frac{1+\delta}{1+\beta}\}$ & -/2 & -/$\min\{2.457,1+\frac{1+\delta}{1+\beta}\}$  & -/2  & {\bf line} \\
  
    \hline
    {\bf Spatial} & 2/$\min\{2.41,2+\delta\}$ & 2/2 & 2/$\min\{2.457,2+\delta\}$  & 2/2  & {\bf arbitrary} \\
    \toprule 
     & & {\bf Multiple servers} & & &\\
     \toprule
     Original$^3$ 
     & $1+\Omega(\frac{1}{k})/1+O(\frac{\log k}{k})$ & -/- & -/- & -/- &line \\
    \hline
    Original$^3$ 
    & 2/2.41 & -/- & -/- & -/-  &arbitrary \\
    \hline
    {Temporal} & -/$\min\{2.04,1+\frac{1}{\beta'\delta'}\}$ & -/$\min \{2,$ & -/$\min\{4,1+\frac{3}{2\beta'\delta'}\}$  & -/$\min \{3,1+\frac{2}{\beta'\delta'}\}$  &{ line} \\
    & & $1+\frac{2}{\max\{2,\beta'\delta'\}}\}$ & & & \\
  
    \hline
    {Temporal} & 2/$\min\{2.41,$  & 2/2 & 2/$\min\{2.457,2+\delta'\min\{\gamma',1\}\}$  & 2/2  &{arbitrary} \\
    &$2+\delta'\min\{\gamma',1\}\}$&&&& \\
   
   \toprule\toprule
    {\bf Spatial} & -/$\min\{2.04,2+\frac{\delta}{\gamma}\}$ & -/2 & -/$\min\{2.457,2(1+\delta)\}$  & -/2  & {\bf line} \\
  
    \hline
    {\bf Spatial} & 2/$\min\{2.41,2(1+\delta)\}$ & 2/2 & 2/$\min\{2.457,2(1+\delta)\}$  & 2/2  & {\bf arbitrary} \\
   
    \bottomrule
    \end{tabular}
\caption{A summary of previous and proposed results for both nomadic and homing versions of {\sc oTsp} and {\sc oDarp} (uncapacitated) for $k\geq 1$ servers. `$X/Y$' means that $X$ and $Y$ are respectively the lower and upper bound on competitive ratio; `-' denotes non-existence of the respective lower/upper bound for the respective problem on the respective metric. We have $\rho=\frac{a}{|\mathcal{T}|}$, $\alpha=\frac{\Theta}{D}$, $\delta=\frac{\Delta}{D}$,  $\beta=\frac{\min\{L,R\}}{D}$, $\gamma=\frac{\max\{L,R\}}{D}$, $\delta'=\frac{\Delta'}{D}$, $\beta'=\min\{1,\frac{t_{max}}{\Delta'}\}$, and $\gamma'=\frac{\max_{1\leq j\leq k}|\mathcal{T}_j|}{D}$  with $a$ being the disclosure time, $\mathcal{T}$ being the TSP tour of the requests in $\sigma$, $\Theta$ being the lookahead time, $D$ being the diameter of the metric space, $L,R$ being the length of the left and right side of the line segment from origin $o$,  $\Delta,\Delta'$, respectively, being the spatial locality and temporal locality, $\mathcal{T}_j$ being the TSP tour of the $j$-th server for the requests in $\sigma$, and $t_{max}$ being the maximum release time among the requests in $\sigma$. 
}
    \label{table:summary}
    \vspace{-3mm} 
\end{table*}

Some investigations in the literature attempted to answer this question. 
Particularly, 
\cite{JailletWagner2006} proposed the {\em disclosure} model in which request $r_i$ with release time $t_i$ is known to the online algorithm $OL$ at time $t_{ia}=t_i-a$ for some constant $a>0$. 
Knowing about $r_i$ releasing at $t_i$ at any time prior to $t_i$ may help $OL$ to plan such that $r_i$ can be served as soon as it is released. That is, the difference $t_i'-t_i$ might get smaller compared to that in the original  online model. 
Indeed, they established the improved competitive ratio for homing {\sc oTsp} (see Table \ref{table:summary}). 
\cite{AllulliABL08} proposed the {\em lookahead} model in which $OL$ knows at any time $t\geq 0$ all the requests with release time $t +\Theta$ (with $\Theta$ the lookahead time). 
They established the improved competitive ratio for both homing and nomadic {\sc oTsp} (see Table \ref{table:summary}). 
Notice that, although named differently, the {\em disclosure} and {\em lookahead} models are essentially equivalent with respect to their working principle. Recently,  \cite{guragain2025temporal} proposed a {\em temporal locality} model (written temporal in Tables \ref{table:summary} and \ref{table:model-comparison}) in which the time interval  between two consecutive (set of) requests in the system is known to $OL$ a priori. They considered the case of uniform time interval and established competitive lower and upper bound results. 

In this paper, we propose a new model, which we call {\em spatial locality} model (written spatial in Tables \ref{table:summary} and \ref{table:model-comparison}), that trades the time locality aspect in \cite{guragain2025temporal} to space locality.  Particularly, in the spatial locality model, the online algorithm $OL$ has the advanced information about the radius at which new  request(s) arrive from the current position of server(s). In other words, suppose at time $t\geq 0$,  a server is  currently at position $p$ in a metric space. If the spatial locality is $\Delta=\delta D, 0\leq \delta< 1$, with $D$ being the diameter of the metric space, then new request(s) arriving in the system  at time $t$ will have source position(s) within radius $\Delta$ from $p$.  In the temporal locality model, time notion is captured meaning that requests arriving at time $t$ may arrive in the metric space from the current position of server, whereas in the spatial locality model, distance notion is captured meaning that requests may arrive at any time as long as it arrives within certain radius of the current position of server. 

Interestingly, the spatial locality model falls back to the original online model if $\Delta=D$ (i.e., $\delta=1$), since new request(s) may come anywhere in the considered metric space.  Moreover, the spatial locality model forgoes the concept of time used in the  previous disclosure, lookahead, and temporal locality models and considers only the concept of distance (comparison in Table \ref{table:model-comparison}). Additionally, like temporal locality, the spatial locality model is weaker than the disclosure and lookahead models as it does not know anything about requests except their locality w.r.t. the current position of server(s). 

\begin{table}[!t]
\footnotesize{
\centering
\begin{tabular}{ll}
\toprule
{\bf Model} & {\bf Characteristic}  \\ 
\toprule
Original &  Requests arriving at time $t_i$ are known at $t_i$ \\
\hline
Lookahead & Requests arriving upto $t_i$ are known at $t_i-a$\\
\hline
Disclosure  & Requests arriving upto $t_i+a$  are known at $t_i$\\
\hline
Temporal & Subsequent request arrival time interval $\Delta$  \\ 
\hline
{\bf Spatial} & Request arrival radius $\Delta'$ from server's position  \\ 
\bottomrule
\end{tabular}
\caption{Comparing online routing models.
\label{table:model-comparison}
}
}
\vspace{-4mm}
\end{table}


As discussed previously, the spatial locality model captures better real-world applications, such as Uber/Lyft rides, than disclosure, lookahead, and temporal locality models. Typically, Uber/Lyft tend to accept ride requests arriving from their vicinity and the riders tend to prefer rides that are closest to them so that their wait/ride time is minimized.  Furthermore, there exist applications which tend to fix spatial locality within which they provide service. If the service provider happens to be a mobile entity then that (mobile) service provider resembles the proposed spatial locality model.

In this paper, our goal is to  quantify the power of the advance information given by the spatial locality model. We consider explicitly the case of  fixed ``maximum amount'' of spatial locality, i.e., at any time $t\geq 0$, the distance  between (at least) a server and newly issued request(s) at time $t$ is at most $\Delta\geq 0$ from the current position of server(s).

Note that, following the literature, we do not focus on {\em runtime complexity} (how quickly the designed algorithms compute the solution) of our algorithms. 
Instead, we focus on the quality of the solution obtained by the online algorithms (measured w.r.t. total time needed to serve all the requests) compared to the solutions obtained by the offline algorithms. 
We note here that our solutions can have runtime complexity polynomial with a small factor increase in competitive ratios. 

\vspace{2mm}
\noindent{\bf Contributions.} 
We consider sequential and non-sequential arrival of requests separately. For the sequential arrival, we prove $(1+\delta)$ competitive ratio, $0\leq \delta (=\Delta/ D)\leq 1$, for both nomadic {\sc oTsp} and {\sc oDarp} on both arbitrary and line metric. For the non-sequential arrival,
we have the following four contributions. The first three contributions consider single server. 
The contributions are outlined and compared in Table \ref{table:summary}. 
The numbers in superscript in Table \ref{table:summary} refer to the following papers: $^1$\cite{AllulliABL08} $^2$\cite{JailletWagner2006} $^3$\cite{bonifaci2009-TCS}$^4$\cite{Bjelde20} $^5$\cite{lipmann2003line}$^6$\cite{AusielloFLST01} $^7$\cite{BirxDS19}$^8$\cite{BaligacsDSW23} $^9$\cite{AscheuerKR00}$^{10}$\cite{feuerstein2001line} $^{11}$\cite{guragain2025temporal}.

Notice that our results are established in the form $\min\{X,Y\}$ with $X$ and $Y$, respectively, being the competitive ratio in the original online and spatial locality models. Depending on $\delta$ (which is known a priori due to the spatial locality model), we approach either runs the original online model algorithm that achieves $X$  or our algorithm that achieves $Y$. If the original model algorithm runs (e.g., when $\delta\geq 0.41$ in arbitrary metric), the spatial model algorithm is never executed and hence $X$ and its analysis holds. If our spatial locality algorithm runs, then the original online model algorithm never runs and our proof for $Y$ holds.

\begin{itemize}
\item We prove a lower bound result -- no deterministic algorithm can achieve better than 2-competitive ratio for both nomadic and homing {\sc oTsp} (which applies directly to both nomadic and homing {\sc oDarp}),  independently of the amount of spatial locality $\Delta=\delta D$. 
({\bf Section \ref{section:lower}})
\item We prove $\min\{2.41,2+\delta\}$-competitive ratio for nomadic {\sc oTsp} on arbitrary metric. The best previously known competitive ratio in the original online model is $1+\sqrt{2}=2.41$ and the spatial locality $\Delta$ such that  $\delta=\Delta/D<\sqrt{2}-1$ provides the improved competitive ratio. 
For an illustration of the benefit, consider a square map of $30 \times 30$ with diameter $D=30\sqrt{2}$. For $\Delta <  0.41 \times 30\sqrt{2}$, our online algorithm for the spatial model provides a better competitive ratio than the online algorithm in the original online model. 

For the line metric, we establish the $\min\{2.04,1+\frac{1+\delta}{1+\beta}\}$-competitive ratio for nomadic {\sc oTsp}, where $\beta=\frac{\min\{L,R\}}{D}$ and $L,R$ are the lengths of the left and right ends of the line segment of length $D$ from the origin $o$. The best previously known competitive ratio in the original model is 2.04 (which is also the lower bound) and our result is an improvement for $\delta=\Delta/D<0.04$ for any value of $0\leq \beta \leq \frac{1}{2}$. 
For an illustration of the benefit, consider a line metric with origin at the either end. In this case $\beta=0$ and hence $\Delta < 0.04 D$ provides improved competitive ratio.  When origin is not at either end, $\beta>0$ and hence the benefit is more visible with increasing $\beta$. If origin at the center, then $\beta=1/2$ and, when $\beta=1/2$, the equation  $1+\frac{1+\delta}{1+\beta}$ reduces to $1+\frac{2}{3}(1+\delta)=1.67+\frac{2}{3}\delta$. Therefore, any $\delta<0.55$ (in other words $\Delta< 0.55 D$) provides improved competitive ratio. ({\bf Section \ref{section:tsp}})

For homing {\sc oTsp}, there is a 2-competitive algorithm in the original model which applies directly in the  spatial locality model.  

\item 
We prove $\min\{2.457,2+\delta\}$-competitive ratio for nomadic {\sc oDarp} on arbitrary metric.
The best previously known competitive ratio in the original model is $2.457$ and  the spatial locality $\Delta$ such that  $\delta=\Delta/D<0.457$ provides the improved competitive ratio.
On line metric, we prove the $\min\{2.457,1+\frac{1+\delta}{1+\beta}\}$-competitive ratio for nomadic {\sc oDarp}. The best previously known competitive ratio in the original model is 2.457 (the lower bound is 2.0585) even for line metric and our result is an improvement for $\delta=\Delta/D<0.457$ for any $0\leq \beta \leq \frac{1}{2}$. When $\beta>0$, the benefit becomes more visible. 
({\bf Section \ref{section:darp}})

For homing {\sc oDarp}, there is a 2-competitive algorithm in the original model which applies directly in the spatial locality model.

\item We consider $k>1$ servers and establish competitive ratios for nomadic {\sc oTsp} and {\sc oDarp} on arbitrary as well as line metric. 
({\bf Section \ref{section:kserver}})
\end{itemize}

Our results exhibit 
that irrespective of the metric space, smaller spatial locality provides better competitive ratios compared to the original model. This is in contrast to the lookahead and disclosure models where larger lookahead and disclosure times were providing better results (refer Table \ref{table:summary}). Additionally, this is also in contrast to the temporal model in which larger locality provides better results for line metric.


In this paper we consider symmetric cost, that is $dist(a,b)=dist(b,a)$ for two request locations $a,b$, homogeneous vehicles, i.e., same $\Delta$ and, at any time $t\geq 0$, the requests coming at time $t$ are within $\Delta$ of the location of server at time $t$. Developing online algorithms considering asymmetric cost, heterogeneous vehicles with different $\Delta$, and requests arriving at distance larger than $\Delta$ are left as important directions for future work.

\vspace{2mm}
\noindent{\bf Challenges and Techniques.} Consider first the case of line metric and single server. Since requests come within $\Delta$ distance from the current position of server, it may be the case that the requests come on both left and right sides of the server. Since server has to pick a direction to serve, the requests on the one side remain outstanding as new request may arrives on the direction server is currently serving. This increases the waiting time of the requests that are on the other side.   Therefore, the server needs to switch direction to serve the requests that were not enroute. As soon as the server switches direction, requests might arrive on that side at which server is not enroute. We overcome this difficulty by  doing the switch only when there is no outstanding request enroute in the current direction. We show that this switching mechanism is indeed able to provide our claimed bounds in the line metric. For arbitrary metric, the requests might come within the circular area of $\Delta$ radius from the current position of server. Here, we were able to show that even in such situations, the request that is closest to the (distinguished) origin can be served by thew server from its current position with additional traversal cost only at most $\Delta$. This proof is highly non-trivial (see Lemma \ref{lemma:delta}).
We then extend these results to multiple servers.

\vspace{2mm}
\noindent{\bf Previous Work.} 
The best previously known lower and upper (competitive ratio) bounds on both {\sc oTsp} and {\sc oDarp}  are given in Table \ref{table:summary} for both line and arbitrary metrics under different models with/without clairvoyance (namely, original, lookahead, disclosure, and temporal locality). 

We first discuss literature on {\sc oTsp}. 
In the original model, {\sc oTsp} was first considered by 
\cite{AusielloFLST01} in which they established tight (competitive ratio)  of 2/2 (lower/upper) on arbitrary metric and  1.64/1.75 on line metric for homing {\sc oTsp}. For nomadic {\sc oTsp}, they provided the lower bound of 2 on line metric and upper bound of $\frac{5}{2}$ on arbitrary metric.   \cite{lipmann2003line} improved the upper bound to $1+\sqrt{2}$ for nomadic {\sc oTsp} on arbitrary metric.  
On line metric, 
\cite{Bjelde20}
provided tight bound of $2.04/2.04$ for nomadic {\sc oTsp} and for homing {\sc oTsp} improved the upper bound to 1.64  matching the lower bound.  Considering multiple servers ($k>1$), 
\cite{bonifaci2009-TCS} provided lower/upper bounds of $1+\Omega(\frac{1}{k})/1+O(\frac{\log k}{k})$ and $2/(1+\sqrt{2})$ for nomadic {\sc oTsp} on line and arbitrary metric, respectively.

In the lookahead model, 
\cite{AllulliABL08} provided an upper bound of $1+\frac{2}{\alpha}$ for both nomadic and homing {\sc oTsp} on line and lower/upper bounds of 2/$\max\{2,1+\frac{1}{2}(\sqrt{\alpha^2+8}-\alpha)\}$ and 2/2 for nomadic and homing {\sc oTsp}, respectively, on arbitrary metric. In the disclosure model, 
\cite{JailletWagner2006} provided an upper bound of $(2-\frac{\rho}{1+\rho})$ for homing {\sc oTsp}.  The lookahead and disclosure models did not consider the problems under multiple servers. 
Recently, \cite{guragain2025temporal} studied the temporal locality model and provided results for $k\geq 1$ servers in both line and arbitrary metric. They also established a lower bound on competitive ratio of 2 for arbitrary metric. Notice that this model captures request arrival time duration aspect but not request location meaning that request(s) may arrive anywhere in the metric. Our spatial locality model captures request location aspect but not time meaning that request(s) may arrive anytime with the only restriction that the request(s) arrives with in radius $\Delta$ of the current server(s) position.

We now discuss literature on {\sc oDarp}. 
In the original model, for homing {\sc oDarp} lower/upper bounds of 2/2 exist on both line and arbitrary metric. For nomadic {\sc oDarp}, the best previously known lower/upper bounds are 2.0585/2.457 \cite{BaligacsDSW23}. {\sc oDarp} was not studied with multiple servers 
 and also not in the lookahead and disclosure models.  
\cite{guragain2025temporal} provided results for $k\geq 1$ in the temporal locality model.

There is also an extensive literature, e.g., \cite{AusielloFLST01,AscheuerKR00,feuerstein2001line}, on solving {\sc oTsp} and {\sc oDarp} with the objective of minimizing the {\em sum of completion times}, e.g., minimize $\sum_{i=1}^m t_i'$, 
in contrast to our considered objective of minimizing the maximum completion time, $\max_{1\leq i\leq m} t_i'$.
In this paper, we focus on minimizing maximum completion time. We acknowledge that minimizing sum of completion times is an important direction that  demands techniques substantially different from the techniques developed in this paper and hence  left for future work.

A distantly related model is of {\em prediction} \cite{GouleakisLS23} which provides the online algorithm $OL$ with the predicted locations beforehand (which may be erroneous) of future requests arriving over time. The arrival time is not predicted in the prediction model. Therefore, the predication model is different from the lookahead, disclosure, temporal, and spatial locality models  since these models (lookahead, disclosure, temporal locality, and spatial locality) do not predict request locations. 
\section{Model}
\label{section:model}
We consider the online (real-time) model where time is divided into steps. We also assume that a time step is equivalent to unit distance and vice-versa. This way, we can measure cost of the designed algorithm based on number of time steps (or number of distance units).  
We consider a sequence  $\sigma=r_1,\ldots,r_m$ of $m$ requests; $m$ is not known a priori.
Every request $r_i=(t_i,e_i,d_i)$ is a triple, where $t_i\geq 0$ is the {\em release time}, $e_i$ is the {\em source} location or point, and $d_i$ is the {\em destination} location.
In {\sc oTsp}, $e_i$ and $d_i$ coincide and hence they can be considered as a single point $e_i$. All the information about $r_i$: $t_i,e_i,d_i$, and its existence is revealed only at time $t_i$. 
We call this (real-time) online model  {\em original}. 
The lookahead and disclosure models \cite{AllulliABL08,JailletWagner2006} assume that the request $r_i$ releasing at time $t_i$ is known at time $t_i-a, a\geq 0$. 

Our spatial locality  model is formally defined as follows. 

\begin{definition}[{\bf spatial locality}]
\label{definition:spatiallocality}
    An online algorithm $OL$ has spatial locality $\Delta\in \mathbf{N}$ at time $t\geq 0$, if and only if request(s) arriving at time $t$ has source location $e_i$ such that $|pos_{OL}(t)-e_i|\leq \Delta$, where $pos_{OL}(t)$ is the position of a server at time $t$.  
\end{definition}

We assume that the execution starts at time $0$. We consider $k\geq 1$ servers $s_1,\ldots, s_k$, initially positioned at a distinguished location called origin $o$. 
The servers can move with unit speed  
traveling in one time step unit distance. 
Each request $r_i$ is called  {\em served} if  at least a server moves first to its source location $e_i$ and then to its destination location $d_i$. 
 To {\em serve} $r_i$, a server has to visit the locations $e_i,d_i$, but not earlier than $t_i$, and $e_i$ has to be visited first and then $d_i$ (if $d_i$ is visited first and then $e_i$, then this does not serve the request). 

Let $(A,B)$ denote the shortest path between two points $A$ and $B$. Let $dist(A,B)$ denote the length of $(A,B)$. 
Following the literature \cite{BienkowskiKL21-ICALP,feuerstein2001line,KrumkePPS03-TCS}, we consider 
arbitrary metric space $\mathcal{M}$ which  
satisfies the following properties: 
\begin{itemize}
    \item [(i)] {\em definiteness:} for any point $x\in \mathcal{M}, dist(x,x)=0$, 
    \item [(ii)] {\em symmetry:} for any two distinct points $x,y\in \mathcal{M}, dist(x,y)=dist(y,x)$, and 
    \item [(iii)] {\em triangle inequality:} for any three distinct points $x,y,z\in \mathcal{M}$, $dist(x,z)+dist(z,y)\geq dist(x,y)$. 
\end{itemize}
Arbitrary metric $\mathcal{M}$ becomes line metric $\mathcal{L}$ when  $dist(x,z)+dist(z,y)=dist(x,y)$ for any three distinct points $x,y,z\in \mathcal{M}$ with $z$ in between $x$ and $y$ in the line segment connecting $x$ and $y$.
 The servers, origin $o$, and  requests in $\sigma$ are all on $\mathcal{M}$ (on $\mathcal{L}$ in case of line metric). Notice that, in the asymmetric setting, the symmetry property above does not hold, i.e., for any two distinct points $x,y\in \mathcal{M}, dist(x,y)\neq dist(y,x)$. Designing online algorithms under such asymmetric setting is left for future work.

The {\em completion time} of request $r_i=(t_i,e_i,d_i)\in \sigma$ is the time $t_i'\geq t_i$ at which $r_i$ has been served; $r_i$ remains {\em outstanding} for the duration $t_i'-t_i$. 
Given $\sigma$, a {\em feasible schedule} for $\sigma$ is a sequence of moves of the server(s) such that all the requests in $\sigma$, arriving following the spatial locality model (Definition \ref{definition:spatiallocality}), are served. $OL(\sigma)$  denotes the maximum completion time of an online algorithm $OL$ for serving the requests in $\sigma$ i.e., $$OL(\sigma)=\max_{1\leq i\leq m} t_i'.$$ $OPT(\sigma)$ is  the optimal completion time of an optimal algorithm $OPT$ for serving the requests in $\sigma$ that arrive following the spatial locality model (Definition \ref{definition:spatiallocality}).  $OPT(\sigma)$ has two components: (i) the largest release time $t_i$ among the requests in $\sigma$, if at least one request has $t_i>0$ (ii) the length of the optimal TSP tour $\mathcal{T}$ connecting the requests in $\sigma$ starting from origin $o$.
Notice that if a request $r$ arrives at time $t> 0$, then it cannot be served before $t$, making $t$ the lower bound for $r$. Additionally, even if all requests in $\sigma$ come at time $t=0$, a request needs $\mathcal{T}$ time to be served. 
Therefore, 
$$OPT(\sigma)\geq \max\left\{\max_{1\leq i\leq m} t_i,|\mathcal{T}|\right\}.$$


The online algorithm $OL$ is said to be $c$-{\em competitive} if $$OL(\sigma) \leq c\cdot OPT(\sigma)$$ for any $\sigma$. The {\em competitive ratio} of 
$OL$ is the smallest $c$ such that $OL$ is $c$-competitive. 



\subsection{Online Algorithm for the Sequential Arrival}
In case of the sequential arrival of the requests in $\sigma$ satisfying spatial locality (Definition \ref{definition:spatiallocality}), we provide an online algorithm $OL$ that achieves competitive ratio $(1+\delta), 0\leq \delta (=\Delta/D) \leq 1,$ for both nomadic {\sc oTsp} and nomadic {\sc oDarp}. 
The {\em sequential arrival} is defined as follows: a request $r_{j}=(t_j,p_j)$ arrives  after the current request $r_i=(t_i,p_i)$ has been served. Any arrival that does not follow this sequential definition is called {\em non-sequential arrival}. 

Suppose request $r_{j}=(t_j,p_j)$ arrives  after the request $r_i=(t_i,p_i)$ have been served. 
        Suppose the server is at location $s_i$ after $r_i$ is served. Due to the spatial locality model, $|p_j-s_i|\leq \delta D$. 
        Let $\sigma_j\subseteq \sigma$ be the set of requests in $\sigma$ that have arrived so far.
        Therefore, server must serve $r_j$ (equivalently $\sigma_j$) by time $$OL(\sigma_j)\leq t_j+|p_j-s_i|\leq t_j+ \delta D.$$ An offline algorithm would have cost $$OPT(\sigma)\geq OPT(\sigma_j)\geq \max\{t_j,|p_j-s_i|\}.$$ Therefore,
        
        $$\frac{OL(\sigma_j)}{OPT(\sigma)}\leq \frac{t_j+ \delta D}{\max\{t_j,|p_j-s_i|\}}\leq 1+\delta.$$

        We describe now why this $1+\delta$ competitive ratio gradually decreases toward 1. Due to the sequential arrival for the request $r_k$ that arrives after $r_j$, $t_k\geq t_j+|p_j-s_i|$. 
        Similarly, the lower bound now becomes $\max\{t_k,|p_k-s_j|\}$.
        Since both nominator and denominator are increasing with the arrival of each new request, the ratio gets monotonically smaller over time (especially after $t>D$).
We summarize the result in the following theorem. 

\begin{theorem}[{\bf nomadic {\sc oTsp} and {\sc oDarp} on arbitrary and line metric}]
\label{theorem:sequential}
    There exists a $(1+\delta)$-competitive algorithm, $0\leq \delta (=\Delta/D)\leq 1$ for nomadic {\sc oTsp} and {\sc oDarp} on any metric (line or arbitrary) when requests arrive sequentially.
\end{theorem}

The significance of Theorem \ref{theorem:sequential} is that it circumvents the lower bound of 2 in competitive ratio we prove in Theorem \ref{theorem:lower} considering the non-sequential case. Notice that the upper bounds in non-sequential case apply to the sequential case but the opposite is not true.  Therefore, the manuscript in the following considers the non-sequential case and establishes interesting lower and upper bound results. We first start with a lower bound.

\section{Lower Bound}
\label{section:lower}
The first natural direction is to study the impact of spatial locality on the solutions to both {\sc oTsp} and {\sc oDarp}. 
We prove that no online algorithm for {\sc oTsp} (nomadic and homing) can be better than 2-competitive in arbitrary metric, especially when requests are non-sequential. Interestingly, our lower bound holds for any value of spatial locality $\Delta$. 

Although we are able to establish the same lower bound for both nomadic and homing versions of {\sc oTsp}, their implication is substantially different. For the homing {\sc oTsp}, 
an optimal 2-competitive online algorithm in the original model exists \cite{AusielloFLST01}, i.e., spatial locality has no impact. Instead, for the nomadic {\sc oTsp}, there is a lower bound of $2.04$ exists in the original model for line metric \cite{Bjelde20} (which directly applies to arbitrary metric). Additionally, the best previously known online algorithm achieves the competitive ratio of $1+\sqrt{2}=2.41$ \cite{lipmann2003line} in the original model for nomadic {\sc oTsp} in  arbitrary metric. We will show in later sections that spatial locality is indeed useful for nomadic {\sc oTsp} typically when the locality is small. 


\begin{theorem}[{\bf lower bound}]
\label{theorem:lower}
No deterministic algorithm for homing/nomadic {\sc oTsp} (and {\sc oDarp}) can be better than $2$-competitive (with no assumption of arrival), independently of the spatial locality $\Delta$.  
\end{theorem}

\begin{proof}
    Consider a star graph $G=(V,E)$ consisting of $N+1$ nodes; a center node $v_0$ and $N$ periphery nodes $v_1,\ldots,v_N$ (see Fig.~\ref{fig:lower_bound}). 
    Each periphery node $v_i$ is connected to the center node $v_0$ by an edge $e_i=(v_0,v_i)$ of length $\Delta$. Let $A$ be any algorithm for homing {\sc oTsp} or nomadic {\sc oTsp}  on $G$ with spatial locality $\Delta$. Let $\Delta$ be such that   $N\mod{\Delta} =0$. 

    Let $t_0=\Delta^2$.
    Until $t_0$, there is no request.
    At time $t_0$, $N$ requests are presented, one on each periphery node (left of Fig.~\ref{fig:lower_bound}). 
    At time $t_0+\Delta$, a request has been served by $A$ and $N-1$ requests are still waiting to be served (middle of Fig.~\ref{fig:lower_bound}).  At time $t_0+2\Delta$, a new request is presented on same vertex on which the request was served at $t_0+\Delta$ (right of Fig.~\ref{fig:lower_bound}). Therefore, at time $t_0+2\Delta$, there is again $N$ requests, one on each vertex $v_i$ of $G$. Continue releasing each requests in the interval of $2\Delta$ time steps until time $t_0+2N\Delta$. 

    \begin{figure}
        \centering
        \includegraphics[scale = 0.17]{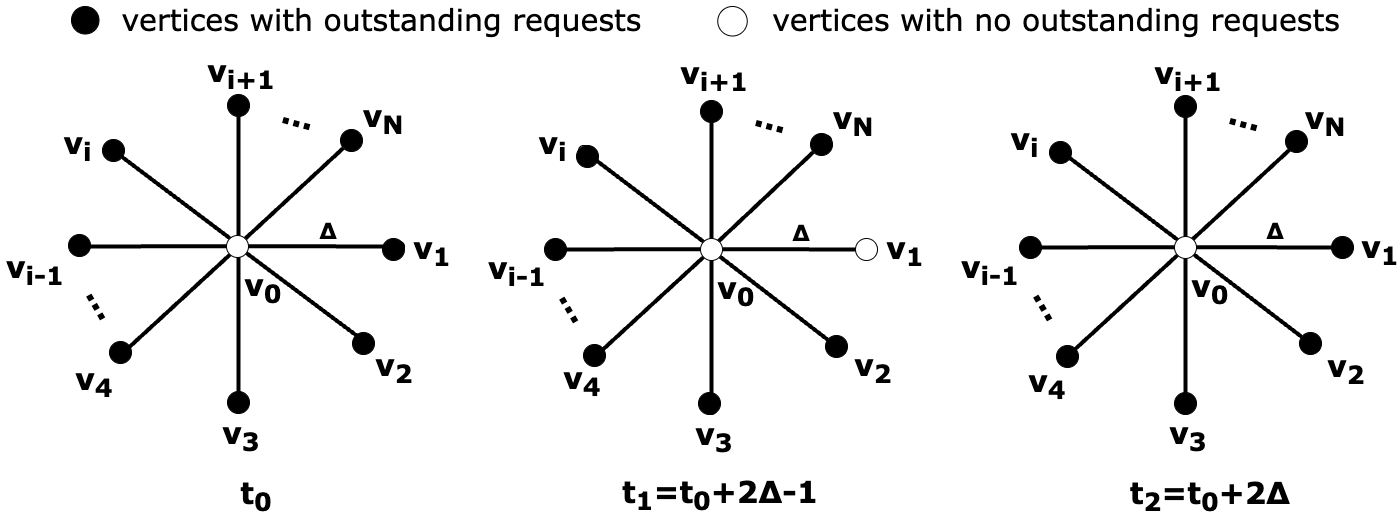}
        \caption{An illustration of lower bound construction in the spatial locality model: ({\bf left}) $N$ requests arrive at time $t_0=\Delta^2$ on $N$ periphery nodes; no request arrives during $[0,t_0)$. ({\bf middle}) A request at some node, say $v_1$, is served at time $t_0+\Delta$ and it remains empty from $t_0+\Delta+1$ until $t_0+2\Delta-1$. ({\bf right}) A new request arrives at the empty node ($v_1$ in the figure) at time $t_0+2\Delta$.}
        \label{fig:lower_bound}
        \vspace{-3mm}
    \end{figure}
    
    Let $t_f={\Delta}({\Delta}+2N)$. At $t_f$, there are exactly $N$ requests in $N$ periphery nodes of $G$ waiting to be served.  Suppose after $t_f$, no new request is presented. That is, in total $2N$ requests have been presented in $G$ from $t_0$ until $t_f$. At time $t_f$, since $A$ still needs to serve $N$ requests, it cannot finish serving them before time $t_f+2{\Delta}N=\Delta({\Delta}+4N)$. 

    We now consider an offline adversary $OPT$. We show that $OPT$ can complete serving all $2N$ requests in no later than time $\Delta^2+2\Delta N$. Consider the vertex request released at time $t_0+2\Delta$.  $OPT$ waits until $t_0+\Delta$ at $v_0$, then reaches the first request at time $t_0+2\Delta$ and serve the $2$ requests on the node by time $t_0+2\Delta$. Continuing this way, $OPT$ finishes not later than time $t_0 + 2\Delta N$ serving all $2N$ requests. For homing {\sc oTsp}, $OPT$ needs additional $\Delta$ time steps to return to $v_0$ after serving the last request. 

    Therefore, the lower bound on competitive ratio becomes
    $$\frac{A(\sigma)}{OPT(\sigma)}\geq \frac{\Delta({\Delta}+4N)}{{\Delta}(\Delta+2N)}.$$
 The ratio $\frac{A(\sigma)}{OPT(\sigma)}$ becomes arbitrarily close to 2 for $N\gg \Delta$. 
\end{proof}

\section{Single-server Online TSP}
\label{section:tsp}
We first present an algorithm for nomadic {\sc oTsp}  that achieves competitive ratio $\min\{2.04,1+\frac{1+\delta}{1+\beta}\}$ in line metric $\mathcal{L}$.   
We then  present an algorithm that achieves competitive ratio $\min\{2.41,2+\delta\}$ in arbitrary metric $\mathcal{M}$. 
%
For homing {\sc oTsp}, there is a 2-competitive algorithm in the original model \cite{AusielloFLST01} which directly applies to the spatial locality model in both line and arbitrary metric.

\subsection{Algorithm on Line Metric}
\label{section:alg-lm-otrp}
The pseudocode is  in Algorithm \ref{algorithm:line-metric}. 
Server $s$ is initially on origin $o$. 
Let $pos_{OL}(t)$ denote the position of server $s$ at any time $t\geq 0$. At $t=0$, $pos_{OL}(0)=o$.
Let $S$ denote the set of outstanding requests. When $S=\emptyset$, the server $s$ stays at its current position  $pos_{OL}(t)$. As soon a new request arrives,
Algorithm \ref{algorithm:line-metric} does the following: Run the original model algorithm of \cite{Bjelde20} which gives competitive ratio 2.04 or run our algorithm that takes spatial locality into account to achieve $(1+\frac{1+\delta}{1+\beta})$ competitive ratio.  The server can do this decision before serving any request since it knows both the parameters $\Delta=\delta D$ and $\beta=\min\{L,R\}/D$. 
Knowing $\delta$ and $\beta$, if $1+\frac{1+\delta}{1+\beta}\geq 2.04$, then we run the original model algorithm of \cite{Bjelde20} that achieves $X=2.04$.  
If $\beta$ and $\delta$ are such that $1+\frac{1+\delta}{1+\beta}<2.04$, we run our approach which achieves $Y=(1+\frac{1+\delta}{1+\beta})$. Therefore the, competitive ratio bound becomes $\min\{X,Y\}=\min\{2.04,1+\frac{1+\delta}{1+\beta}\}.$ 

Our approach proceeds as follows.
Among the requests in $S$, Algorithm \ref{algorithm:line-metric} finds the two extreme positions $L$ and $R$ of the requests in $S$. 
We have two cases: (i) both $L$ and $R$ are on the same side on $\mathcal{L}$ from $pos_{OL}(t)$ (including $pos_{OL}(t)$), (ii) $L$ is on one side and $R$ is on another side on $\mathcal{L}$ from $pos_{OL}(t)$. For Case (i), tour $T$ is constructed connecting $pos_{OL}(t)$ with $L$ or $R$ whichever is farthest. For Case (ii), tour $T$ is constructed connecting $pos_{OL}(t)$ with first $L$ then $R$ if $dist(pos_{OL}(t),L)\leq dist(pos_{OL}(t),R)$, otherwise, tour $T$ is constructed connecting $pos_{OL}(t)$ with first $R$ then $L$. 
Server $s$ traverses the tour $T$ until either a new request arrives or $T$ is traversed. If no new request arrives during the traversal of $T$ and the server $s$ finished traversing $T$, server $s$ stays at its current position $pos_{OL}(t)$ which is the end point of $T$. 
Fig.~\ref{fig:singleserver-line}  illustrates these ideas.

    \begin{figure}[htbp]
    \centerline{
    \includegraphics[scale = 0.27]{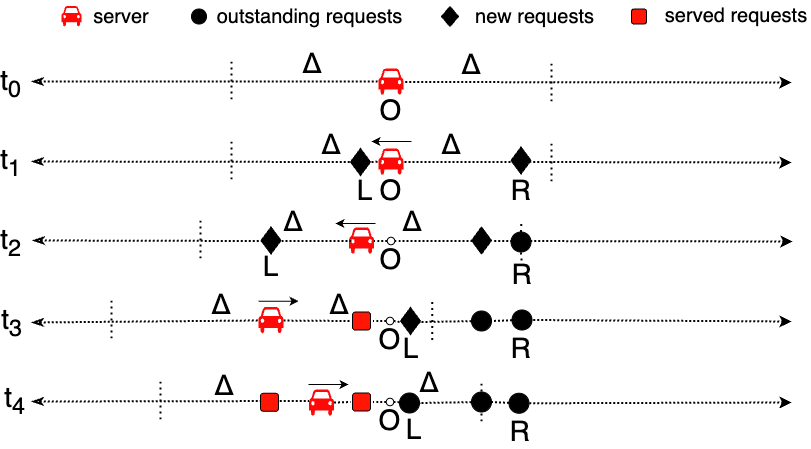}
    }
    \caption{An illustration of Algorithm \ref{algorithm:line-metric} for {\sc oTsp} on line. The $\Delta$ interval on the left and right of the current position of server shows how far from the current position new request would arrive. At $t_0$, server $s$ is on $o$, origin. At $t_1$, two requests arrive and $s$ picks left direction. At $t_2$, since $s$ finds new outstanding request enroute on left, it continues traversing left. At $t_3$, there is no outstanding request on the left and hence $s$ changes direction to right. At $t_4$, it is continuing serving requests enroute on the right direction.}
    \label{fig:singleserver-line}
    \vspace{-3mm}
    \end{figure}

\begin{algorithm}[t!]
\footnotesize
\begin{algorithmic}[1]
\STATE $o\leftarrow$ origin where server $s$ resides initially \\
\STATE $pos_{OL}(t)\leftarrow$ the current position of server $s$ on line $\mathcal{L}$ at time $t\geq 0$; $pos_{OL}(t=0)=o$\\
\IF{($1+\frac{1+\delta}{1+\beta}\geq 2.04$)}
\STATE Run the original model algorithm of \cite{Bjelde20}
\ELSE
\IF{new request(s) arrives}
\STATE $S\leftarrow $ the set of outstanding requests including the new request(s)\\
\STATE $L\leftarrow$ the farthest position $e_L$ among the requests in $S$ on $\mathcal{L}$\\
\STATE $R\leftarrow$ the farthest position $e_R$ among the requests in $S$ on $\mathcal{L}$ on the opposite side of $L$\\
\IF{both $L$ and $R$ are on the same side on $L$ from $pos_{OL}(t)$, including  $pos_{OL}(t)$}

\STATE $T\leftarrow$ a tour connecting $pos_{OL}(t)$ with $L$ or $R$ whichever is farthest 

\ELSE

\IF{$dist(pos_{OL}(t),L)\leq dist(pos_{OL}(t),R)$}
\STATE $T\leftarrow$ a tour connecting $pos_{OL}(t)$ with $L$ and then $L$ with $R$
\ELSE
\STATE $T\leftarrow$ a tour connecting $pos_{OL}(t)$ with $R$ and then $R$ with $L$
\ENDIF
\ENDIF
\STATE server $s$ traverses $T$ until a new request arrives or $T$ is traversed;  if $T$ is traversed, then stay at $pos_{OL}(t)$ \\ 
\ENDIF
\ENDIF

\caption{Single-server algorithm for nomadic {\sc oTsp} on line metric $\mathcal{L}$ with spatial locality $\Delta$}
 \label{algorithm:line-metric}
\end{algorithmic}
\end{algorithm}

\vspace{2mm}
\noindent{\bf Analysis of the Algorithm.} 
Let $L',R'$ be the two extreme positions of the requests in  $\sigma\cup \{o\}$, meaning that all $m$  requests and origin $o$ have positions between $L'$ and $R'$ on $\mathcal{L}$ (inclusive). 
Let $\mathbb{L}=|o-L'|$ and $\mathbb{R}=|o-R'|$. We have that the length of the line segment $D=\mathbb{L}+\mathbb{R}=\min\{\mathbb{L},\mathbb{R}\}+\max\{\mathbb{L},\mathbb{R}\}$. 
Let $t_{max}$ be the arrival time of the request $r_{max}=(t_{max},e_{max})$ in $\sigma$ released last, i.e., there is no other request $r'=(t',e')$ with $t'>t_{max}$. Let $t_{min}$ be the arrival time of the request $r_{min}=(t_{min},e_{min})$ in $\sigma$ released first, i.e., there is no other request $r''=(t'',e'')$ with $t''<t_{min}$.
%
%
%
%
%
We first establish correctness.

\begin{lemma}
\label{lemma:correctness-line}
Algorithm \ref{algorithm:line-metric} serves all the requests in $\sigma$. 
\end{lemma}
\begin{proof}
We prove this by contradiction. Support a request $r_i=(e_i,t_i)$ is not served. It must be the case that server $s$ has not reached $e_i$. By construction, we have that $e_i$ is either on the left or right of the current position $pos_{OL}(t)$ of server $s$. As soon as a request arrives, $s$ moves towards that request. For two or more requests arriving at the same time, server $s$ moves toward one. The direction is changed as soon as there is no outstanding request enroute, which happens in the worst case the extreme end of $\mathcal{L}$. After change in direction, the server $s$ does not stop until all outstanding requests are served. Therefore, $e_i$ must be visited in the server $s$'s traversal during left or right direction, hence a contradiction. 
\end{proof}

We now establish the competitive ratio bound.
\begin{theorem}[{\bf nomadic {\sc oTsp} on line metric}]
\label{theorem:oTsp-line}
    Algorithm \ref{algorithm:line-metric} with spatial locality $\Delta=\delta D$ is $\min\{2.04,1+\frac{1+\delta}{1+\beta}\}$-competitive for nomadic {\sc oTsp} defined on an interval of length $D$, where $0\leq \delta\leq 1$ and $0\leq \beta=\frac{\min\{\mathbb{L},\mathbb{R}\}}{D}\leq \frac{1}{2}$. 
    
\end{theorem}
\begin{proof}
Consider the input instance $\sigma$. Suppose all $m$ requests are released at $t=0$.
We have that $|\mathcal{T}|=2\min\{\mathbb{L},\mathbb{R}\}+\max\{\mathbb{L},\mathbb{R}\}=D+\min\{\mathbb{L},\mathbb{R}\}=(1+\beta) D$, since $\beta=\min\{\mathbb{L},\mathbb{R}\}/ D$.
Algorithm \ref{algorithm:line-metric} finishes serving requests in $\sigma$ in $|\mathcal{T}|$  time, i.e. $OL(\sigma)\leq |\mathcal{T}|$. Any optimal algorithm $OPT$ also needs at least $|\mathcal{T}|$ time, i.e., $OPT(\sigma)\geq |\mathcal{T}|$. Therefore,
Algorithm \ref{algorithm:line-metric} is 1-competitive. 

Therefore, we focus on the case of not all requests released at $t=0$. We first establish the lower bound. Consider $r_{max}=(t_{max},e_{max})$, the request in $\sigma$ released last, with $t_{max}>0$. 
Since $r_{max}$ cannot be served before $t_{max}$ by any algorithm,  $OPT(\sigma)\geq t_{max}$. Furthermore, $$OPT(\sigma)\geq |\mathcal{T}|\geq D+\min\{\mathbb{L},\mathbb{R}\}=(1+\beta)D$$ even when not all requests arrive at time $t=0$.  

We now establish the upper bound for Algorithm \ref{algorithm:line-metric}. Notice that we know  $\delta$ due to the spatial locality model and $\beta$ due to the location of origin in line. If Line 3 in Algorithm \ref{algorithm:line-metric} is true, i.e., $1+\frac{1+\delta}{1+\beta}\geq 2.04$,  then we run the algorithm of \cite{Bjelde20} and hence the upper bound analysis of \cite{Bjelde20} gives $X=2.04$ competitive ratio. 
However, if Line 3 is not true, i.e., $1+\frac{1+\delta}{1+\beta}<2.04$, we prove $Y=1+\frac{1+\delta}{1+\beta}$ competitive ratio, and hence overall competitive ratio is $\min\{X,Y\}$.
The proof is as follows: at $t_{max}$, 
the length of the tour $T$ computed by Algorithm \ref{algorithm:line-metric} cannot be larger than $\Delta+D$. 
Consider the configuration at $t_{max}$. Since the new requests must be released within $\Delta$ distance from $pos_{OL}(t)$, there is at least one side (left or right) from $pos_{OL}(t)$ which has all outstanding requests within $\Delta$ distance of $pos_{OL}(t)$. 
Therefore, after $t_{max}$, each outstanding request waits for at most $\Delta+D$ time units before being served. Therefore, $OL(\sigma)\leq t_{max}+\Delta+D.$

Combining the above results, 
$$\frac{OL(\sigma)}{OPT(\sigma)}\leq \frac{t_{max}+\Delta+D}{\max\{t_{max},(1+\beta)D\}}.$$

We have two cases: (a) $t_{max}\geq (1+\beta)D$ or (b) $t_{max}<(1+\beta)D$. 
For Case (a), $D\leq \frac{t_{max}}{1+\beta}$ and since $\Delta=\delta D$,
\begin{equation}
\begin{split}
OL(\sigma) & \leq  \left(1+\frac{\delta t_{max}}{(1+\beta)t_{max}} + \frac{t_{max}}{(1+\beta)t_{max}}\right) OPT(\sigma)\\
& =  (1+\frac{1+\delta}{1+\beta}) OPT(\sigma).
\end{split}
\end{equation}

For Case (b), replacing $t_{max}$ with $(1+\beta)D$, we obtain
\begin{equation}
\begin{split}
OL(\sigma) & \leq  \left(1+\frac{\delta D}{(1+\beta)D}+\frac{D}{ (1+\beta)D}\right) OPT(\sigma)\\
& =(1+\frac{1+\delta}{1+\beta}) OPT(\sigma).
\end{split}
\end{equation}
The ratio $OL(\sigma)/OPT(\sigma)=1+\frac{1+\delta}{1+\beta}$ is $Y$ and hence we obtain $\min\{X,Y\}$ competitive ratio for Algorithm \ref{algorithm:line-metric} on line metric. 
\end{proof}

\subsection{Algorithm on Arbitrary Metric}
\label{section:alg-am-otrp}
The pseudocode of the algorithm is given in Algorithm \ref{algorithm:arbitrary-metric}. Server $s$ is initially on $o$, the origin. In Algorithm \ref{algorithm:arbitrary-metric}, server $s$ serves the requests as follows. Let $pos_{OL}(t)$ be the current position of $s$ at time $t$; $pos_{OL}(0)=o$.  Whenever (at least) a new request arrives at time $t$, $s$ constructs a tour $T$ as follows. 
If $\delta\geq 0.41$, it runs the procedure similar to \cite{lipmann2003line} which achieves $X=2.41$ competitive ratio. If $\delta<0.41$, then our approach is run which gives competitive ratio $Y=2+\delta$. Therefore, the overall competitive ratio becomes $\min\{X,Y\}=\min\{2.41,2+\delta\}$. Our approach is  discussed below.
Let $S_t$ be the set of requests in $\sigma$ arrived until time $t$. 
Let $\mathcal{T}_t$ be the minimum cost TSP tour of the requests in $S_t\cup \{o\}$ with $o$ being the one end of the tour. Let $e$ be the position of the first outstanding request after $o$ in  $\mathcal{T}_t$. 
Let $\mathcal{T}_t(o,e)$ be the part of the tour $\mathcal{T}_t$ from origin $o$ to $e$. 
We have that $T:=(pos_{OL}(t),e) \cup \mathcal{T}_t\backslash \mathcal{T}_t(o,e)$, i.e., $T$ connects the current position $pos_{OL}(t)$  of server $s$ with $e$ and then    $\mathcal{T}_t$ without considering its sub-tour $\mathcal{T}_t(o,e)$ from origin $o$ to $e$.    
Server $s$ then traverses the tour $T$ starting from $pos_{OL}(t)$ until a new request arrives while traversing $T$ or until $T$ is traversed. If $T$ is traversed before any request arrives, server $s$ stops at the endpoint of $T$ until a new request arrives. If a new request arrives at some time $t'>t$ before finish traversing $T$, $s$ again computes $T$ at $t'$ as described above and traverses the tour $T$. Fig.~\ref{fig:singleserver-arbitrary} illustrates these ideas.

\begin{algorithm}[bt!]
\footnotesize
\begin{algorithmic}[1]
\STATE $o\leftarrow$ origin where server $s$ resides initially \\ 
\STATE $pos_{OL}(t)\leftarrow$ the current position of server $s$ on metric $\mathcal{M}$ at time $t\geq 0$; $pos_{OL}(0)=o$
\IF{$\delta=\Delta/D\geq 0.41$}
\STATE Run the original model algorithm of \cite{lipmann2003line}
\ELSE
\IF{new request(s) arrives at time $t$}
\STATE $S_t \leftarrow$ the set of requests in $\sigma$ arrived until $t$
\STATE $\mathcal{T}_t\leftarrow$ the minimum cost TSP tour that connects the positions of the requests in $S_t\cup\{o\}$ with $o$ being the one endpoint of the tour \\
\STATE $e\leftarrow$ position of the outstanding request in $S$ that comes immediately after $o$ in $\mathcal{T}_t$\\
\STATE $(pos_{OL}(t),e)\leftarrow$ the line segment that connects the current position $pos_{OL}(t)$ of server $s$ with $e$\\ 
\STATE $\mathcal{T}_t(o,e)\leftarrow$ the part of   tour $\mathcal{T}_t$ from $o$ to $e$
\STATE $T:= \{(pos_{OL}(t),e)\} \cup \mathcal{T}_t\backslash\mathcal{T}_t(o,e)$
\STATE server $s$ traverses tour $T$ until either it finishes or a new request arrives  
\ENDIF
\ENDIF
\caption{Single-server algorithm for nomadic {\sc oTsp} on arbitrary metric $\mathcal{M}$ with spatial locality $\Delta$}
 \label{algorithm:arbitrary-metric}
 \end{algorithmic}
\end{algorithm}

\vspace{2mm}

\begin{figure*}[htbp]
    \centerline{
    \includegraphics[scale = 0.25]{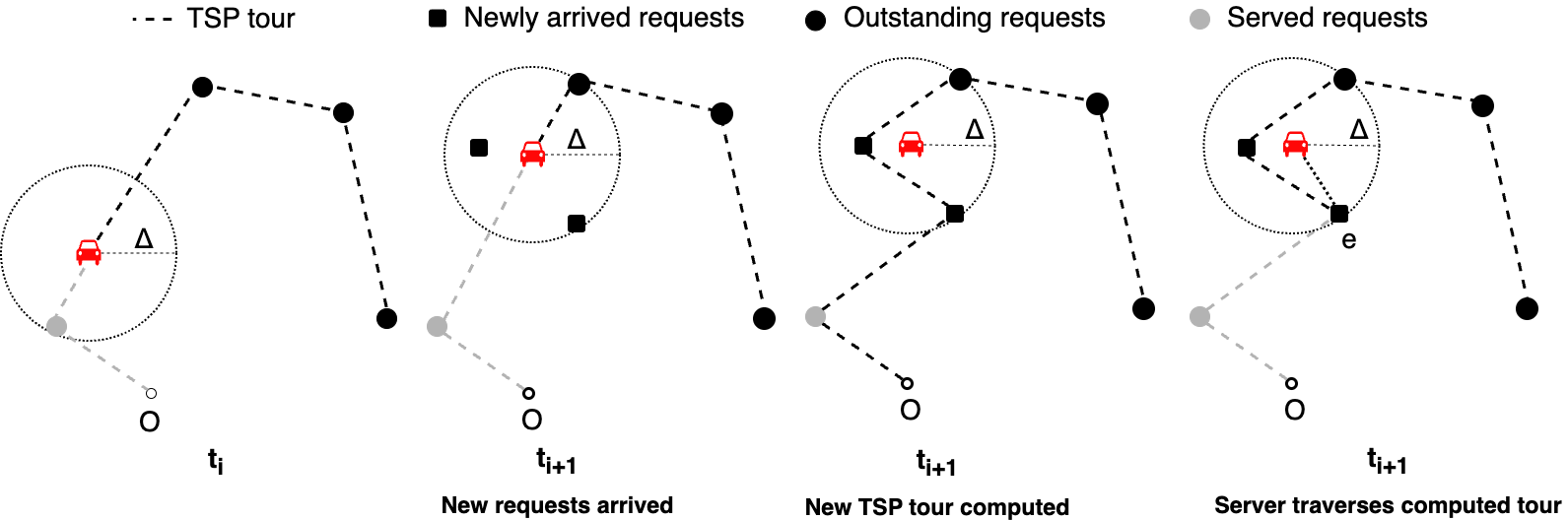}
    }
    \caption{An illustration of Algorithm \ref{algorithm:arbitrary-metric} for {\sc oTsp} on arbitrary metric.  ({\bf left}) Suppose at time $t_i$ the first set of requests arrive; server $s$ computes the TSP tour $\mathcal{T}_{i}$ of the requests and starts to traverse the tour starting from origin $o$ ({\bf second left}) While traversing the tour, at time $t_{i+1}>t_i$ two new requests arrived within spatial locality $\Delta$ of the current position $pos_{OL}(t_{i+1})$ of server ({\bf second right}) At $t_{i+1}$, server $s$ computes a new TSP tour $\mathcal{T}_{i+1}$ of all the requests arrived so far starting from origin $o$ ({\bf right}) server $s$ computes the part of tour $T$ (which connects its current position $pos_{OL}(t_{i+1})$ to the position of the first outstanding request $e$ in $\mathcal{T}_{i}$ from $o$ and not considering the part of $\mathcal{T}_{i+1}$ from $o$ to $e$) that needs to be traversed and starts to traverse $T$ from  $pos_{OL}(t_{i+1})$.} 
    \label{fig:singleserver-arbitrary}
    \vspace{-3mm}
    \end{figure*}

\noindent{\bf Analysis of the Algorithm.}
We first prove correctness. 

\begin{lemma}
\label{lemma:correctness-arbitrary}
Algorithm \ref{algorithm:arbitrary-metric} serves all the requests in $\sigma$. 
\end{lemma}

\begin{proof}
    We prove this by contradiction. Suppose a request $r_i=(e_i,t_i)$ is not served. When $r_i$ arrives at $t_i$, it must be outstanding. At $t_i$, the server $s$ computes the TSP tour $T$ that visits the locations of all outstanding requests. Server $s$ does not stop until $T$ is fully traversed and there is no outstanding request. Since $r_i$ was outstanding at $t_i$ and after, it must have been served by server $s$ running Algorithm \ref{algorithm:arbitrary-metric}, a contraction.   
\end{proof}

We now establish the competitive ratio.
We start with the following lemma which is crucial in the proof.

\begin{lemma}
\label{lemma:delta}
Suppose (at least) a request $r_i=(e_i,t_i)$ arrives at time $t_i$. Let $S_i$ be the set of requests in $\sigma$ arrived so far (served and outstanding). Let $\mathcal{T}_i$ be the TSP tour connecting the locations of the requests in $S_i$ with $o$ being the one end of $\mathcal{T}_i$.  Let $e$ be the location of the first outstanding request after $o$ in $\mathcal{T}_i$. Let $S_\Delta$ be the set of outstanding requests within $\Delta$ radius from the current position $pos_{OL}(t_i)$ of server $s$.  It must be that $e$ is among the locations of the requests in $S_\Delta$. 
\end{lemma}
\begin{proof}
Let $S_{i-1}$ be the set of requests in $\sigma$ (served and outstanding) arrived until time $t_{i-1}$.
Consider the TSP tour $\mathcal{T}_{i-1}$ that connects the locations of the requests in $S_{i-1}$ with $o$ being the one end of $\mathcal{T}_{i-1}$. Let $t_i>t_{i-1}$. Suppose (at least) a new request arrives at $t_i$ and no new request arrives after $t_{i-1}$ and before $t_i$.  Let $\mathcal{T}_i$ be the TSP tour connecting the locations of the requests in $S_i$ with $o$ being the one end of $\mathcal{T}_i$. 
Let $\mathcal{T}_i(o,e)$ be the part of $\mathcal{T}_i$ from origin $o$ to $e$ (the location of the first outstanding request in $\mathcal{T}_i$ after $o$). Let $S^{o,e}_{served}$ be the set of served requests in $\sigma$ that are on $\mathcal{T}_i(o,e)$, except $o$ and $e$.  

We consider two cases: (i) $S^{o,e}_{served}=\emptyset$ (ii) $S^{o,e}_{served}\neq \emptyset$. 
In Case (i), an outstanding request must be immediately after $o$ in  $\mathcal{T}_i$. Let $e'$ be the location of the first request from  $\mathcal{T}_{i-1}$ (served or outstanding).
Since $e'$ is not the first outstanding request in $\mathcal{T}_{i-1}$, it must be the case that $dist(o,e)<dist(o,e')$ with $e$ being the location of a request that arrived at time $t_i$. Since the locations of all the requests arrived at $t_i$ are within $\Delta$ distance from $pos_{OL}(t_i)$ of server $s$, $e$ must be among the locations of the requests in $S_\Delta$.

In Case (ii), there must be the case that for each $r'\in S^{o,e}_{served}$, $dist(o,e')<dist(o,e)$, since otherwise each $r'\in S^{o,e}_{served}$ would have been ordered in  $\mathcal{T}_i$ after $e$. We now show that $e$ must be among the locations of the requests in $S_\Delta$. Let $S^{arr}_i$ be the set of requests released at time $t_i$. We have that $S_i:=S_{i-1}\cup S^{arr}_i$. 
Recall that each request in $S^{arr}_i$ has location with in $\Delta$ from $pos_{OL}(t_i)$. 
We have two sub-cases for the location $e$ of the first outstanding request in $\mathcal{T}_i$ after $o$: (i) $e\in S_{i-1}$ (ii) $e\in S^{arr}_i$. In sub-case (i), $dist(pos_{OL}(t_i),e)\leq \Delta$ and $e$ must be the location of the request in $S_{i-1}$ such that it is ordered before any request in $S^{arr}_i$, otherwise, $e$ must be the location of a request in $S^{arr}_i$. In sub-case (ii), it is immediate that $dist(pos_{OL}(t_i),e)\leq \Delta$. Therefore, $e$ must be among the locations of the outstanding requests in $S_{\Delta}$, since $S_{\Delta}$ contains only the outstanding requests with in $\Delta$ radius of the current position  $pos_{OL}(t_i)$ of the server $s_i$ at time $t_i$.
\end{proof}

\begin{theorem}[{\bf nomadic {\sc oTsp} on arbitrary metric}]
\label{theorem:trp-arbitrary-ub}
    Algorithm \ref{algorithm:arbitrary-metric} with spatial locality $\Delta=\delta D$ is $\min\{2.41,2+\delta\}$-competitive for nomadic {\sc oTsp} on arbitrary metric with diameter $D$. 
\end{theorem}
\begin{proof}
Consider the input instance $\sigma$. We know $\delta$. Therefore,
if $\delta\geq 0.41$, we run the algorithm of \cite{lipmann2003line} which provides $X=2.41$ competitive ratio. For $\delta<0.41$, we prove that our approach provides $Y=2+\delta$ competitive ratio. Combining these two bounds, we have overall competitive ratio $\min\{X,Y\}=\min\{2.41,2+\delta\}$. 

The analysis of our approach is  as follows to obtain $Y=2+\delta$ competitive ratio.
Let $\mathcal{T}$ be the minimum cost TSP tour connecting the locations of the requests in $\sigma\cup \{o\}$ with $o$ being the one end of the tour. 
%
%
Suppose (at least) a request $r_i=(e_i,t_i)$ arrives at time $t_i$
Let $S_i$ be the set of requests in $\sigma$ arrived until $t_i$. $S_i$ contains both served and outstanding requests. 
Let $\mathcal{T}_i$ be the TSP tour connecting the locations of the requests in $S_i$ with $o$ being the one end of $\mathcal{T}_i$.
We show that at $t_i$, the length of the tour $T$ computed by Algorithm \ref{algorithm:arbitrary-metric} to visit the outstanding requests in $S_i$ cannot be larger than $|\mathcal{T}_i|+\delta D$, i.e., $|T|\leq |\mathcal{T}_i|+\delta D$.
The proof is as follows. The location $e$ of the first outstanding request after $o$ in $\mathcal{T}_i$ is within distance $\Delta$ from $pos_{OL}(t_i)$ (Lemma \ref{lemma:delta}). 
Additionally, server $s$ does not need to traverse that part of $\mathcal{T}_i$ from $o$ to $e$. 
Therefore, the total tour length to traverse for $s$ is $T\leq dist(pos_{OL}(t_i),e)\cup \mathcal{T}_i\backslash \mathcal{T}_i(o,e)\leq \Delta + \mathcal{T}_i=\delta D+\mathcal{T}_i$.

Let $r_{max}=(t_{max},e_{max})$ be the request in $\sigma$ released last. Since no new request arrives after $t_{max}$, server $s$ will finish serving all requests within  $\delta D+|\mathcal{T}_{max}|$. 
We have that $|\mathcal{T}_{max}|\leq |\mathcal{T}|$.
Therefore, $OL(\sigma)\leq t_{max}+\delta D+|\mathcal{T}|.$
We have that $OPT(\sigma)\geq t_{max}$. Irrespective of whether all the requests are released at time $t\geq 0$,  $OPT(\sigma)\geq |\mathcal{T}|.$
Therefore,
$$\frac{OL(\sigma)}{OPT(\sigma)}\leq \frac{t_{max}+\delta D+|\mathcal{T}|}{\max\{t_{max},|\mathcal{T}|\}}.$$

Since the diameter is $D$, at least the location of two requests is $D$ apart and hence $|\mathcal{T}|\geq D$. 
We have two cases: (a) $t_{max}\geq |\mathcal{T}|$ or (b) $t_{max}< |\mathcal{T}|$. 
For Case (a),
$$OL(\sigma)\leq \left(\frac{2t_{max}}{t_{max}}+\frac{\delta t_{max}}{t_{max}}\right) OPT(\sigma)=(2+\delta) OPT(\sigma).$$

For Case (b), $$OL(\sigma)\leq \left(\frac{2|\mathcal{T}|}{|\mathcal{T}|}+\frac{\delta |\mathcal{T}|}{|\mathcal{T}|}\right) OPT(\sigma)=(2+\delta) OPT(\sigma).$$

The theorem follows combining the competitive ratio $X=2.41$ in \cite{lipmann2003line} for $\delta\geq 0.41$ and $Y=2+\delta$ for our approach for $\delta<0.41$.
\end{proof}

\section{Single-server Online DARP}
\label{section:darp}

We modify Algorithms \ref{algorithm:line-metric} and \ref{algorithm:arbitrary-metric} to solve {\sc oDarp} on line and aritrary metric, respectively. We consider nomadic {\sc oDarp}. For homing {\sc oDarp}, there is a 2-competitive algorithm in the original model \cite{AscheuerKR00,feuerstein2001line} which directly solves homing {\sc oDarp} in the spatial locality model in both line and arbitrary metric. 

The only change in Algorithms \ref{algorithm:line-metric} and \ref{algorithm:arbitrary-metric} to handle {\sc oDarp} is after visiting the source location of a request its destination location needs to be visited.  Furthermore, the spatial locality needs to be preserved only for the source location, i.e., the destination location for a request arriving at time $t$ can be more than $\Delta$ away from $pos_{OL}(t)$ of server $t$.

\begin{figure*}[!t]
    \centering
    \includegraphics[scale=0.25]{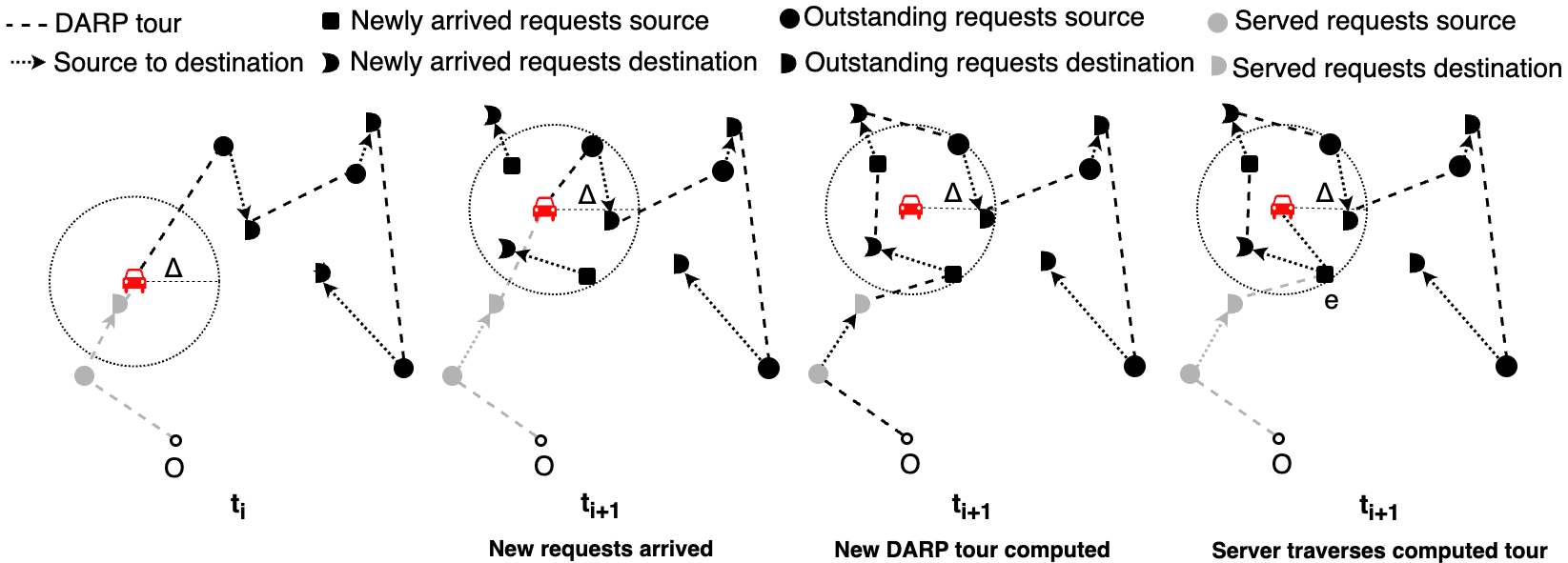}
    \caption{An illustration of Algorithm \ref{algorithm:arbitrary-metric} for {\sc oDarp} on arbitrary metric. The description is analogous to Fig.~\ref{fig:singleserver-arbitrary}.}
    \label{fig:Darp-Arbitrary}
    \vspace{-3mm}
\end{figure*}

\subsection{Algorithm on Line Metric}
\label{section:odarp-lm-1}
The correctness proof is analogous to Lemma \ref{lemma:correctness-line}. We prove the following theorem for the competitive ratio bound.
\begin{theorem}[{\bf nomadic {\sc oDarp} on line metric}]
\label{theorem:trp-line-ub}
    Algorithm \ref{algorithm:line-metric} with spatial locality $\Delta=\delta D$ is $\min\{2.457,1+\frac{1+\delta}{1+\beta}\}$-competitive for nomadic {\sc oDarp} defined on an interval of length $D$, where $0\leq \delta\leq 1$ and $0\leq \beta=\frac{\min\{\mathbb{L},\mathbb{R}\}}{D}\leq \frac{1}{2}$. 
\end{theorem}
\begin{proof}
Consider the input instance $\sigma$. Suppose all $m$ requests are released at $t=0$. Let $|\mathcal{T}|$ be the length of the optimal TSP tour that serves all the requests in $\sigma$. 
Algorithm \ref{algorithm:line-metric} finishes serving requests in $\sigma$ in $|\mathcal{T}|$  time, i.e. $OL(\sigma)\leq |\mathcal{T}|$. Any optimal algorithm $OPT$ also needs at least $|\mathcal{T}|$ time, i.e., $OPT(\sigma)\geq |\mathcal{T}|$. Therefore,
Algorithm \ref{algorithm:line-metric} is 1-competitive. 

Suppose not all requests are released at $t=0$. 
When $(1+\frac{1+\delta}{1+\beta})\geq 0.457$, we run the algorithm of \cite{BaligacsDSW23} and we obtain the $X=2.457$ competitive ratio. However, when $1+\frac{1+\delta}{1+\beta}< 0.457$, our approach is used which we analyze in the next paragraph, obtaining $Y=1+\frac{1+\delta}{1+\beta}$ competitive ratio. Therefore, the overall competitive ratio becomes $\min\{X,Y\}=\min\{2.457,1+\frac{1+\delta}{1+\beta}\}$. 

The analysis of our approach is involved since it needs to consider not just the source location but also the destination location and on what side destination location lies from the current position of server is impacts the analysis differently.

Consider the request released last $r_{max}=(t_{max},e_{max},d_{max})$. 
Server $s$ cannot start serving $r_{max}$ before $t_{max}$. Therefore, $OPT(\sigma)\geq t_{max}$. 
We now introduce the notion of increasing and non-increasing requests. Consider $\mathcal{L}$ with one end $L$ (denoting left) and another end $R$ (denoting right). Consider a request $r_i=(t_i,e_i,d_i)$ with $e_i$ being closer to $L$ than $d_i$ and let's call $r_i$ {\em increasing}.  If there is another request with  $d_i$ being closer to $L$ than $e_i$, then that request will be {\em non-increasing}.   The notation can be analogous in reference to $R$ as well.  

We now prove lower and upper bounds based on whether the requests in $\sigma$ are all increasing requests or there is at least one request that is non-increasing. 
If all requests are increasing, we prove that $OL(\sigma)\leq t_{max}+\Delta+D$. 
The proof is as follows. Let $\sigma_j$ be the set of request(s) released at $t_j$. We have that for any request $r_j\in \sigma_j$, $|pos_{OL}(t_j)-e_j|\leq \Delta$, i.e., the source location of each request in $\sigma_j$ is within $\Delta$ distance from the current location $pos_{OL}(t_j)$ of server $s$ at time $t_j$. Therefore, server $s$ can reach the farthest location of the requests in $\sigma_j$ at time at most $t_j+\Delta$. Since all the requests are increasing, and when no further requests arrive (i.e., after $r_{max}$), server $s$ finishes serving all the outstanding requests in next $D$ time steps. 
For the optimal cost, pick $L$ or $R$ which happens to be the source location $e_j$ of some  request $r_j$. Since all requests are increasing, server $s$ must first go to $L$ or $R$ (that is the source location) and then to $R$ or $L$ serving all the requests with minimum cost. Therefore, as in Theorem \ref{theorem:oTsp-line}, $OPT(\sigma)\geq D+\min\{\mathbb{L},\mathbb{R}\}\geq (1+\beta)D$. 

Combining the above results, 
$$\frac{OL(\sigma)}{OPT(\sigma)}\leq \frac{t_{max}+\Delta+D}{\max\{t_{max},(1+\beta)D\}}.$$

Analysis as in Theorem \ref{theorem:oTsp-line} gives 
$$OL(\sigma) =  (1+\frac{1+\delta}{1+\beta}) OPT(\sigma).$$

We now consider the case of $\sigma$ containing non-increasing requests. 
We first argue that the lower bound $OPT(\sigma)\geq D+\min\{\mathbb{L},\mathbb{R}\}\geq (1+\beta)D$ for the all increasing request case also applies to the case of $\sigma$ containing at least one non-increasing request. This is because both endpoints $L$ and $R$ must be visited by  server $s$ starting from origin $o$, which takes at least  $D+\min\{\mathbb{L},\mathbb{R}\}$, even when all requests in $\sigma$ are released at time $t=0$.

We now prove the upper bound  for Algorithm \ref{algorithm:line-metric} for the case of $\sigma$ containing non-increasing requests. Consider the current position $pos_{OL}(t_k)$ of server $s$ at any time $t_k\geq 0$ and let the direction of $s$ is toward $L$. Let $\sigma_k$ be the set of requests released at time $t_k$. We have that $|pos_{OL}(t_k)-e_k|\leq \Delta$, where $e_k$ is the source point of any request in $\sigma_k$. Server $s$ will reach the source point of any request in $\sigma_k$ at time $t_k+\Delta$. We have two situations for each request $r_k\in \sigma_k$:  (i) $d_k$ is toward $L$ (ii) $d_k$  is toward $R$. In Case (i), $s$ serves request $r_k$ in next $|e_k-d_k|$ time steps. Therefore, if no further request(s) arrive at time $>t_k$, then all the requests in $\sigma_k$ satisfying Case (i) will be served in $t_k+\Delta+ed^{max}_k$, where $ed^{max}_k$ is the maximum source-destination location distance of the requests in $\sigma_k$ satisfying Case (i). After that, the remaining outstanding requests (from $\sigma_k$ and all others) must have destination points toward $R$ and they will be served by $s$ in next $D$ time steps. 
Therefore, $OL(\sigma)\leq t_k+\Delta+ed^{max}_k+D.$ If we consider the request sets $\sigma_1,\ldots, \sigma_k$ up to $t_k$, we have $OL(\sigma)\leq \max_{\sigma_k\in \sigma} (t_k+ed^{max}_k)+\Delta+D.$

We now prove the lower bound  $OPT(\sigma)\geq \max_{\sigma_k\in \sigma} (t_k+ed^{max}_k)$. The argument is as follows. Consider the set $\sigma_k$ of requests released at $t_k$. Since $ed^{max}_k$ is the maximum source-destination location distance of the requests in $\sigma_k$, it must be traversed to finish serving the requests in $\sigma_k$ after $t_k$. 

Therefore, 
$$\frac{OL(\sigma)}{OPT(\sigma)}\leq \frac{\max_{\sigma_k\in \sigma} (t_k+ed^{max}_k)+\Delta+D}{\max\{\max_{\sigma_k\in \sigma} (t_k+ed^{max}_k),(1+\beta)D\}}.$$

Analysis as in Theorem \ref{theorem:oTsp-line} gives 
$$OL(\sigma) =  (1+\frac{1+\delta}{1+\beta}) OPT(\sigma).$$
The theorem follows combining the above bound $Y=1+\frac{1+\delta}{1+\beta}$ with the competitive ratio $X=2.457$ in \cite{BaligacsDSW23}. 
\end{proof}

\subsection{Algorithm on Arbitrary Metric}
\label{section:odarp-am-1}
Fig.~\ref{fig:Darp-Arbitrary} illustrates Algorithm \ref{algorithm:arbitrary-metric} for nomadic {\sc oDarp} in arbitrary metric.
The correctness proof is analogous to Lemma \ref{lemma:correctness-arbitrary}. We prove the following theorem for the competitive ratio bound.

\begin{theorem}[{\bf nomadic {\sc oDarp} on arbitrary metric}]
\label{theorem:oDarp}
    Algorithm \ref{algorithm:arbitrary-metric} with spatial locality $\Delta=\delta D$ is $\min\{2.457,2+\delta\}$-competitive for nomadic {\sc oDarp} on arbitrary metric with diameter $D$. 
\end{theorem}
\begin{proof}
Consider the input instance $\sigma$.  When $\delta\geq 0.457$, we again run the algorithm of \cite{BaligacsDSW23}, which gives $X=2.457$ competitive ratio. However, when $\delta< 0.457$, our analysis gives $Y=2+\delta$ competitive ratio. Therefore, the overall competitive ratio becomes $\min\{X,Y\}=\min\{2.457,2+\delta\}$. 

The analysis of our approach for $\delta<0.457$ is as follows. 
Let $r_{max}=(t_{max},e_{max})$ be the request in $\sigma$ released last. 
Let $\mathcal{T}_{Darp}$ be the minimum length tour for the requests in $\sigma$ such that one endpoint of $\mathcal{T}_{Darp}$ is origin $o$ and for each request $r_i$, its $e_i$ and $d_i$ come consecutively in $\mathcal{T}_{Darp}$ exactly once. 
At $t_{max}$, 
the length of the tour $T$ computed by Algorithm \ref{algorithm:arbitrary-metric} cannot be larger than $\delta D+|\mathcal{T}_{Darp}|$.
This is because we can prove Lemma \ref{lemma:delta} for the source locations. Moreover, the TSP tour computed whenever a new request arrives cannot be longer than   $|\mathcal{T}_{Darp}|$. Therefore, 
after $t_{max}$, each outstanding request is served before time $\delta D+ |\mathcal{T}_{Darp}|$, i.e.,  $OL(\sigma)\leq t_{max}+\delta D+|\mathcal{T}_{Darp}|.$
We have that $OPT(\sigma)\geq t_{max}$. Irrespective of whether all the requests are released at time $t\geq 0$,  $OPT(\sigma)\geq |\mathcal{T}_{Darp}|.$ Thus,

$$\frac{OL(\sigma)}{OPT(\sigma)}\leq \frac{t_{max}+\delta D+ |\mathcal{T}_{Darp}|}{\max\{t_{max},|\mathcal{T}_{Darp}|\}}.$$

Since the diameter is $D$ and at least the locations of two requests is $D$ apart, $|\mathcal{T}_{Darp}|\geq D$. 
We have two cases: (a) $t_{max}\geq |\mathcal{T}_{Darp}|$ or (b) $t_{max}< |\mathcal{T}_{Darp}|$. 
For Case (a),
$$OL(\sigma)\leq \left(\frac{2t_{max}}{t_{max}}+\frac{\delta t_{max}}{t_{max}}\right) OPT(\sigma)=(2+\delta) OPT(\sigma).$$

For Case (b), $$OL(\sigma)\leq \left(\frac{2|\mathcal{T}_{Darp}|}{|\mathcal{T}_{Darp}|}+\frac{\delta |\mathcal{T}_{Darp}|}{|\mathcal{T}_{Darp}|}\right) OPT(\sigma)=(2+\delta) OPT(\sigma).$$

The theorem follows combining our $Y=2+\delta$ bound with $X=2.457$ competitive bound in \cite{BaligacsDSW23}.
\end{proof}

\begin{figure*}[!t]
    \centering
    \includegraphics[scale=0.25]{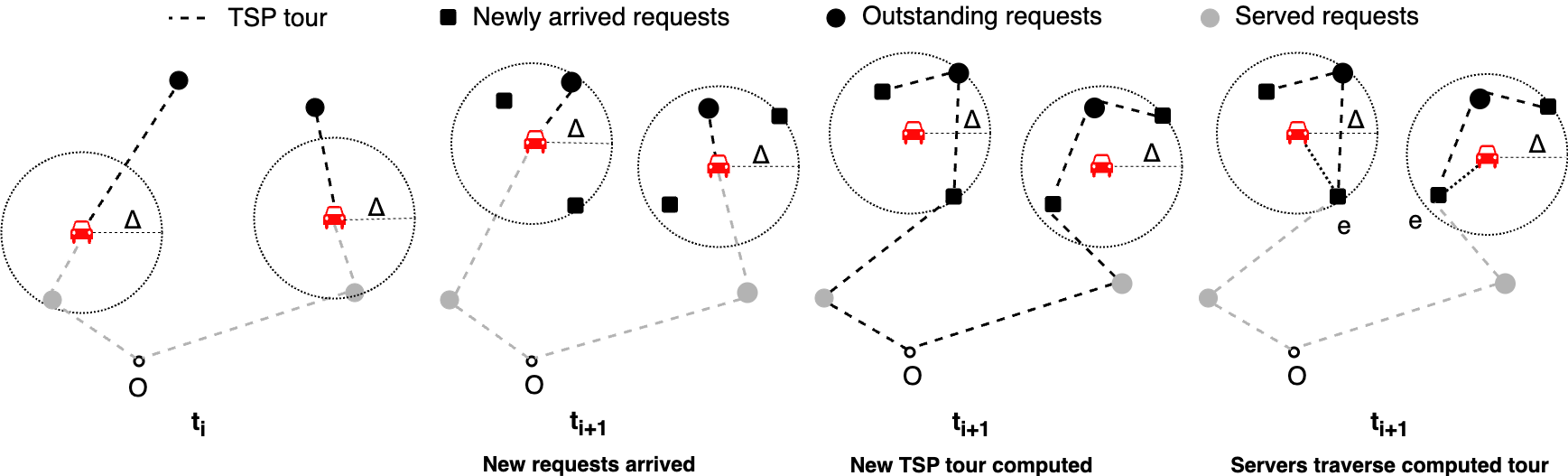}
    \caption{An illustration of parallelized Algorithm \ref{algorithm:arbitrary-metric}  nomadic {\sc oTsp} in arbitrary metric with 2 servers.}
    \label{fig:k-server-Arbitrary}
    \vspace{-3mm}
\end{figure*}

\section{$k>1$ Server Extensions}
\label{section:kserver}
We now discuss how the competitive ratios for $k=1$ server extend to $k>1$ servers for both {\sc oTsp} and {\sc oDarp}. Fig.~\ref{fig:k-server-Arbitrary} illustrates an example run for nomadic {\sc oTsp} in arbitrary metric with 2 servers. 

\begin{theorem}[{\bf nomadic {\sc oTsp} on line metric with $k>1$ servers}]
\label{theorem:tsp-line-k}
    Parallelized Algorithm \ref{algorithm:line-metric} with spatial locality $\Delta=\delta D$ is  $\min\{2.04,2+\frac{\delta}{\gamma}\}$-competitive for nomadic {\sc oTsp} defined on an interval of length $D$ for $k>1$ servers, where $\frac{1}{2}\leq \gamma=\frac{\max\{\mathbb{L},\mathbb{R}\}}{D}\leq 1$. 
\end{theorem}
\begin{proof}
Consider the input instance $\sigma$. Suppose all $m$ requests are released at $t=0$. Having $k\geq 2$ is enough for parallelized Algorithm \ref{algorithm:line-metric} to serve the requests in $\sigma$ with cost $OL(\sigma)\leq \max\{\mathbb{L},\mathbb{R}\}$. Any optimal algorithm $OPT$ also needs at least $\max\{\mathbb{L},\mathbb{R}\}$ time, i.e., $OPT(\sigma)\geq \max\{\mathbb{L},\mathbb{R}\}$. Therefore,
Algorithm \ref{algorithm:line-metric} is 1-competitive. 
Suppose not all requests are released at $t=0$.
Consider $r_{max}=(t_{max},e_{max})$, the request in $\sigma$ released last. 
Since $r_{max}$ cannot be served before $t_{max}$ by any online algorithm,  $OPT(\sigma)\geq t_{max}$.
At $t_{max}$, 
the tour $T$ computed by Algorithm \ref{algorithm:line-metric} for each server $s_j$ cannot be larger than $\Delta+\max\{\mathbb{L},\mathbb{R}\}$. 
Therefore, $OL(\sigma)\leq t_{max}+\Delta+\max\{\mathbf{L},\mathbb{R}\}.$

Combining the above results, 
$$\frac{OL(\sigma)}{OPT(\sigma)}\leq \frac{t_{max}+\Delta+\max\{\mathbb{L},\mathbb{R}\}}{\max\{t_{max},\max\{\mathbb{L},\mathbb{R}\}\}}.$$

Suppose $\gamma=\frac{\max\{\mathbb{L},\mathbb{R}\}}{D}$.
We have that,
$$OL(\sigma)\leq \left(2+\frac{\delta}{\gamma}\right) OPT(\sigma).$$

Instead of our approach above obtaining $Y=2+\frac{\delta}{\gamma}$, since $\delta$ and $\gamma$ are known, we opt to run the parallelized algorithm in \cite{bonifaci2009-TCS} if $(2+\frac{\delta}{\gamma})\geq 2.04$ to obtain $X=2.04$ competitive ratio instead.  Therefore, combining these two competitive bounds, we obtain $\min\{2.04, 2+\frac{\delta}{\gamma}\}$ competitive ratio.  
\end{proof}

\begin{theorem}[{\bf nomadic {\sc oTsp} on arbitrary metric with $k>1$ servers}]
\label{theorem:otsp-arbitrary-ub-k}
    Parallelized Algorithm \ref{algorithm:arbitrary-metric} with spatial locality $\Delta=\delta D$ is $\min\{2.41,2(1+\delta)\}$-competitive for nomadic  {\sc oTsp} on arbitrary metric with diameter $D$ for $k>1$ servers. 
\end{theorem}
\begin{proof}
Let $r_{max}=(t_{max},e_{max})$ be the request in $\sigma$ released last. At $t_{max}$, 
the length of the tour $T_j$ computed by Algorithm \ref{algorithm:arbitrary-metric} for server $s_j$ cannot be larger than $|\mathcal{T}_j|+\delta D$, i.e., $|T_j|\leq |\mathcal{T}_j|+\delta D$. This is because since the spatial locality is $\Delta=\delta D$, the first outstanding request in each tour $T_j$ cannot be more than $\Delta$ away from the current position of at least a server $s_j$ as in Lemma \ref{lemma:delta}.
After $t_{max}$, each outstanding request served by $s_j$ waits for at most $|\mathcal{T}_j|+\delta D$ time units before being served by $s_j$. Therefore, $OL(\sigma)\leq t_{max}+\max_{1\leq j\leq k}|\mathcal{T}_j|+\delta D.$
We have that $OPT(\sigma)\geq t_{max}$. Irrespective of whether all the requests are released at time $t\geq 0$, since $OPT$ must visit all the requests, it must pay at least $\max_{1\leq j\leq k}|\mathcal{T}_j|$, i.e., $OPT(\sigma)\geq \max_{1\leq j\leq k}|\mathcal{T}_j|.$

Combining the above results, 
$$\frac{OL(\sigma)}{OPT(\sigma)}\leq \frac{t_{max}+\max_{1\leq j\leq k}|\mathcal{T}_j|+\delta D}{\max\{t_{max},\max_{1\leq j\leq k}|\mathcal{T}_j|\}}.$$

Since the diameter is $D$, at least the location of two requests is $D$ apart and hence $\max_{1\leq j\leq k}|\mathcal{T}_j|\geq D/2$. 
We have two cases: (a) $t_{max}\geq \max_{1\leq j\leq k}|\mathcal{T}_j|$ or (b) $t_{max}< \max_{1\leq j\leq k}|\mathcal{T}_j|$. 
For both cases,
$$OL(\sigma)\leq (2+2\delta) OPT(\sigma)=2(1+\delta)OPT(\sigma).$$

As in Theorem \ref{theorem:tsp-line-k},  since $\delta$ is known, we opt to run the parallelized algorithm in \cite{bonifaci2009-TCS} if $2(1+\delta)\geq 2.41$ to obtain $X=2.41$ competitive ratio instead of our above competitive bound $Y=2(1+\delta)$. Therefore, combining these two competitive bounds $X$ and $Y$, we obtain $\min\{2.41,2(1+\delta)\}$ competitive ratio.
\end{proof}

\begin{theorem}[{\bf nomadic {\sc oDarp} on arbitrary metric with $k>1$ servers}]
\label{theorem:trp-arbitrary-ub}
    Parallelized Algorithm \ref{algorithm:arbitrary-metric} with spatial locality $\Delta=\delta D$ is $\min\{2.457,2(1+\delta)\}$-competitive for nomadic  {\sc oDarp} on any (arbitrary or line) metric with diameter $D$ for $k>1$ servers. 
\end{theorem}

\begin{proof}
Consider the input instance $\sigma$. 
The proof is analogous to Theorem \ref{theorem:otsp-arbitrary-ub-k}.
Let $r_{max}=(t_{max},e_{max})$ be the request in $\sigma$ released last. 
Let $\mathcal{T}^j_{Darp}$ be the minimum length tour server $s_j$ for the requests in $\sigma$ such that one endpoint of $\mathcal{T}^j_{Darp}$ is origin $o$ and for each request $r_i$, its $e_i$ and $d_i$ come consecutively in $\mathcal{T}^i_{Darp}$ exactly once. 
At $t_{max}$, 
the length of the tour $T_j$ computed by Algorithm \ref{algorithm:arbitrary-metric} for each server $j$ cannot be larger than $\delta D+|\mathcal{T}^j_{Darp}|$. This is because since the spatial locality is $\Delta=\delta D$, the first outstanding request in each tour $T_j$ cannot be more than $\Delta$ away from the current position of at least a server $s_j$ as in Lemma \ref{lemma:delta}. 
Therefore, $OL(\sigma)\leq t_{max}+\delta D+ \max_{1\leq j\leq k}|\mathcal{T}^j_{Darp}|.$
We have that $OPT(\sigma)\geq \max\{t_{max},\max_{1\leq j\leq k}|\mathcal{T}^j_{Darp}|\}$.

Combining the above results, 
$$\frac{OL(\sigma)}{OPT(\sigma)}\leq \frac{t_{max}+\delta D+\max_{1\leq j\leq k}|\mathcal{T}^j_{Darp}|}{\max\{t_{max},\max_{1\leq j\leq k}|\mathcal{T}^j_{Darp}|\}}.$$

As in Theorem \ref{theorem:otsp-arbitrary-ub-k}, $\max_{1\leq j\leq k}|\mathcal{T}^j_{Darp}|\geq D/2$. 
Therefore, we obtain $OL(\sigma)\leq 2(1+\delta) OPT(\sigma).$

As in Theorem \ref{theorem:tsp-line-k},  since $\delta$ is known, we opt to run algorithm in \cite{BaligacsDSW23} if $2(1+\delta)\geq 2.457$ to obtain $X=2.457$ competitive ratio instead for $k>1$ servers. Our analysis  $2(1+\delta)<2.457$ gives $Y=2(1+\delta)$ competitive ratio. Therefore, combining these two competitive bounds $X$ and $Y$, we obtain $\min\{2.457,2(1+\delta)\}$ competitive ratio.
\end{proof}

\section{Concluding Remarks}
\label{section:conclusion}
In this paper, we have proposed the new clairvoyance model, called spatial locality, and studied its power on online routing via two fundamental problems  {\sc oTsp} and {\sc oDarp} that have relevant applications in logistics and robotics with $k\geq 1$ servers. 
We proved the simple competitive  ratio of $1+\delta, 0\leq \delta\leq 1,$ for nomadic {\sc oTsp} and {\sc oDarp} in case of sequential arrival of requests.
For any arrival (not necessarily sequential), 
we first established a lower bound of 2-competitive ratio for  both homing and nomadic versions of {\sc oTsp} and {\sc oDarp} in arbitrary metric. 
We then showed that in both arbitrary and line metric the competitive ratios better than the currently known bounds in the original non-clairvoyant model can be obtained when spatial locality is small (for arbitrary metric, $\delta<0.41$). For future work, it will be interesting to consider other online  problems and/or other objective functions such as minimizing sum of completion times.  It will also be interesting to run experimental evaluations against real-world routing traces collected from different applications in logistics and robotics, such as Uber/Lyft, to see how our algorithms and their bounds translate to practice under varying amount of spatial locality and request arrival timings. Finally, it will be interesting to consider asymmetric cost, heterogeneous vehicles, etc., and design competitive algorithms.


\bibliographystyle{named}
\bibliography{references}

\end{document}